\documentclass[11pt, peerreview]{IEEEtran}
\renewcommand{\baselinestretch}{1.}

\usepackage{amssymb}
\usepackage[pdftex]{graphicx}
\usepackage{fancybox}

\newtheorem{definition}{Definition}

\newtheorem{lemma}{Lemma}

\newtheorem{theorem}{Theorem}

\def\unc{unconnected}
\def\con{connected}
\def\gam{{k \choose u} {u(k-u) \choose v} {{k \choose 2} - u(k-u) \choose n - v}}
\begin{document}

\title{Probabilistic Analysis of the Network Reliability Problem on a Random Graph Ensemble}
\author{\authorblockN{Akiyuki Yano and Tadashi Wadayama   } \\[0.3cm]
\authorblockA{Nagoya Institute of Technology\\
{\small email: wadayama@nitech.ac.jp }
\footnote{A part of this work was posted to arXiv (arXiv:1105.5903, 30 May, 2011). }\\
}
}  

\maketitle
\begin{abstract}
In the field of computer science, the network reliability problem for evaluating the network failure probability has been extensively investigated. 
For a given undirected graph $G$, the network failure probability is the probability that edge failures (i.e., edge erasures) make $G$ unconnected. Edge failures are assumed to occur independently with the same probability. The main contributions of the present paper are the upper and lower bounds on the expected network failure probability. We herein assume a simple random graph ensemble that is closely related to the Erd\H{o}s-R\'{e}nyi random graph ensemble. These upper and lower bounds exhibit the typical behavior of the network failure probability. The proof is based on the fact that the cut-set space of $G$ is a linear space over $\Bbb F_2$ spanned by the incident matrix of $G$. The present study shows a close relationship between the ensemble analysis of the expected network failure probability and the ensemble analysis of the error detection probability of LDGM codes with column weight 2.
\end{abstract}

\section{Introduction}
\subsection{Background}
Network systems are ubiquitous in communication networks and power line networks, for example. Designing a reliable network is critical for achieving a system that is robust to unexpected failures. Theoretical treatments, with an abstraction of a real network as an undirected graph, provide insight on robust networks, which will be useful for network design.

The {\em network reliability problem} proposed by Moore and Shannon \cite{REL} 
has been extensively investigated primarily in the fields of computer science and
combinatorics. Assuming that an undirected graph $G$ is given, this undirected graph $G$ can be regarded as an
abstract model of a real network.  For example, in a packet communication scenario, a vertex and an edge represent a router and an communication link, respectively. We also assume that each edge can be broken. If some edges are broken (the event is referred to as {\em edge failure}),  
then the graph may become unconnected.  In a packet communication scenario,
there exist two routers, between which there is no communication route. We refer to such an event 
as a {\em network failure}. Network reliability problems are problems for evaluating the probability of a network failure for a given undirected graph.

Although there are several edge failure models, in the present paper, we adopt the simplest of these. The assumption is as follows. Edge failures occur independently, and the probability of an edge failure is uniformly given by $\epsilon (0 < \epsilon < 1)$. We also assume that all of the vertices are perfectly reliable. In the present paper, we focus on the {\em all-terminal scenario} \cite{ALLTERM}, in which a network is considered 
to be valid only if there exists a path between any two vertices, i.e., the graph representing a network is connected.

Evaluation of the network failure probability for a given undirected graph is known to be a computationally demanding problem. Provan and Ball \cite{NP1} and Valiant \cite{NP2} demonstrated that the network reliability problems are $\#{\cal P}$-complete,
which is a complexity class at least as intractable as ${\cal NP}$. Thus, the existence of polynomial-time algorithms for solving the network reliability problem appears unlikely.
Karger \cite{Poly} presented a randomized polynomial time approximation algorithm for the all terminal network reliability problem. Buzacott \cite{LITTLE1} and Ball and Provan \cite{LITTLE2}, for example, have developed exponential time algorithms for evaluating the exact network failure probability with a smaller exponential order. Instead of exact evaluation, 
Ball and Provan \cite{BALL-PROVAN} and Colbourn \cite{EDGE-PACKING} presented upper bounds and lower bounds for all-terminal network reliability, which can be evaluated in linear time.

\subsection{Main contributions}

In the present paper, we present a probabilistic analysis of all-terminal network reliability problems that is based on a random graph ensemble. The random graph ensemble assumed herein is closely related to the Erd\H{o}s-R\'{e}nyi random graph model
\cite{RANDOMGRAPH}.
The primary contribution of the present paper is the upper and lower bounds on the expected network failure probability.
These upper and lower bounds reveal 
the typical behavior of the network failure probability. The upper bound to be presented in Theorem \ref{upperbound} has 
the following form: 
\begin{equation}
{\sf E}[P_f(G,\epsilon)] \le
\frac{1}{2{{k \choose 2}\choose  n}} \sum_{v=1}^n  \sum_{u=0}^k 
\gam  \epsilon^{v}   
+  \frac{1}{2{{k \choose 2}\choose  n}}  \sum_{u=0}^{k} {k \choose u}  {{k \choose 2}-u(k-u) \choose n  }  - 1,
\end{equation}
where $k, n$ are the number of vertices and edges, respectively. The probability $P_f(G,\epsilon)$ is the network 
failure probability of $G$ under the condition where the edge failure probability is $\epsilon$. 
The expectation is taken over a random graph ensemble of unconnected graphs with $k$-vertices and $n$-edges.
As far as the authors know,
this type of bounds on the expected network failure probability appear to be novel.

In order to derive these bounds, an average cut-set weight distribution is derived. 
The set of cut-set vectors (i.e., the incidence vector of a cut-set) 
of an undirected graph forms a linear space over binary field ${\Bbb F_2}$ spanned by 
an incidence matrix. 
A combinatorial argument  similar to the ensemble analysis of low-density generator matrix (LDGM) code with
column weight 2 is exploited to obtain the average cut-set weight distribution.  
The {\em coding theory perspective} presented in the paper could be useful
not only for the network reliability problem but also for 
several graph problems related to the cut-set structure of a graph.


\section{Preliminaries}
\label{prel}
\subsection{Random graph ensemble}
We herein consider an ensemble (i.e., a probability space) of undirected graphs.
Let $k (k \ge 1)$ be the number of labeled vertices, and let
$n (1 \le n \le k(k-1)/2)$ be the number of labeled edges.
The vertices are labeled from $1$ to $k$, and the edges are labeled from $1$ to $n$.
For any adjacent vertices, only a single edge is allowed.
In the following, $[a,b]$ denotes a set of consecutive integers from $a$ to $b$.
The set $G_{k,n}$ denotes the set of all undirected weighted graphs with 
$k$-vertices and $n$-edges satisfying the above assumption.
For any $G \in G_{k,n}$, the sets of vertices and edges are denoted by $V(G)$ and $E(G)$, respectively.
The cardinality of  $G_{k,n}$ is given by 
$
|G_{k,n} | = n! {{k \choose 2}\choose  n}.
$
Here, we assign the equal probability 
\[
P(G)= \frac{1}{n! {{k \choose 2}\choose  n}}
\]
for $G \in G_{k,n}$ as the probability measure.
Note that the pair $(G_{k,n}, P)$ defines an ensemble of random graphs and 
is  denoted by ${\cal G}_{k,n}$.

\subsection{Cut-set}

For a given undirected graph $G$, a partition $V(G)=V_1 \cup V_2$ satisfying 
$
\emptyset = V_1 \cap V_2
$
is referred to as a {\em cut} of $G$.
The set of edges bridging $V_1$ and $V_2$ is referred to as the {\em cut-set} corresponding to 
a cut $(V_1, V_2)$. The weight of a cut (or a cut-set) is defined as the cardinality of the cut-set and is denoted by $\omega(E)$, where $E$ is a cut-set.

For $G \in G_{k,n}$, the {\em incidence matrix} of $G$, denoted by $M(G) \in \{0,1\}^{k \times n}$,
is defined as follows:
\begin{equation}
M(G)_{i,j} = 
\left\{
\begin{array}{ll}
1, & \mbox{if $i$th vertex connects to $j$th edge}\\
0, & \mbox{otherwise},\\
\end{array}
\right.
\end{equation}
where $M(G)_{i,j}$ is the $(i,j)$-element of $M(G)$.

There are two important properties of $M(G)$.
If $G$ is connected, then the rank (over $\Bbb F_2$) of $M(G)$ is $k-1$.
The row space  (over $\Bbb F_2$) of the incidence matrix $M(G)$ 
corresponds to the set of all possible cut-sets of $G$, which is called a {\em cut space} \cite{Diestel} \cite{CUT}.
Note that the cut space property plays a crucial role in the following analysis. 
The matrix  $M(G)$ can be regarded as a generator matrix of LDGM code with column weight 2 \cite{modern}.
The weight distribution of such an LDGM code can be interpreted as the cut-set weight distribution.

\subsection{Cut-set weight distribution}

We define the cut-set weight distribution of $G$ by
\begin{equation}
B_v(G) = \sum_{E \subset E(G)} \Bbb I [E \mbox{ is a cut-set of }G, \     \omega(E)= v]
\end{equation}
for non-negative integer $v$. The function $\Bbb I[\cdot]$ is an indicator function that takes a value of 1 if the condition is true, otherwise this function takes a value of 0. 

The integer-valued function  
$
A_{u,v}: G_{k,n} \rightarrow \Bbb Z 
$
is defined by
\begin{equation}
A_{u,v}(G) = \sum_{m \in Z^{(k,u)} } \sum_{c \in Z^{(n,v)} }\Bbb I \left[m M(G) = c \right] .
\end{equation}
The set of constant weight binary vectors 
$Z^{(a,b)}$ is defined as
$
Z^{(a,b)} = \{x \in  \{0,1\}^a: w_H(x)= b  \},
$
where the function $w_H(\cdot)$ represents the Hamming weight.  
The following lemma is another representation of the cut-set weight distribution using the
incidence matrix of $G$:
\begin{lemma}\label{bwg}
Assume  $G \in G_{k,n}$. If $G$ is a connected graph,  $B_v(G)$ is given by
\begin{equation}\label{upperbound}
B_v(G) = \frac{1}{2} \sum_{u=0}^k A_{u,v}(G)
\end{equation}
for $u \in [0,k], v \in [0,n]$.

\end{lemma}
\begin{proof}
We denote the row space of the incident matrix of $G$ as ${\sf span}(M(G))$.
The weight distribution $B_v(G)$ can be rewritten as follows:
\begin{eqnarray} \nonumber
B_v(G) 
&=&  \sum_{E \subset E(G)} \Bbb I [E \mbox{ is a cut-set of }G, \     \omega(E)= v]\\ \label{abcd}
&=& \sum_{c \in {\sf span}(M(G)) }  \Bbb I \left[ w_H(c)=v \right]. 
\end{eqnarray}
The second equality is due to the fact that the row space of $M(G)$ is equal to the set of all possible 
cut-set vectors of $G$. Note that, based on the assumption, the rank of $M(G)$ should be $k-1$.
The equality 
\begin{equation} \label{rowspace}
\sum_{c \in {\sf span}(M(G)) } h(c) = \frac {1}{2} \sum_{m \in \{0,1\}^k} \sum_{c \in \{0,1\}^n} h(c)\Bbb I \left[c = m M(G) \right]
\end{equation}
holds for any real-valued function $h: \{0,1\}^n \rightarrow \Bbb R$
because
\begin{equation}
|\{m \in \{0,1\}^k \mid  m M(G) = c\}| = 2
\end{equation}
holds for any $c \in {\sf span}(M(G))$.
Substituting (\ref{rowspace}) into (\ref{abcd}), we obtain 
\begin{eqnarray} \nonumber
B_v(G) 
&=& \frac {1}{2} \sum_{m \in \{0,1\}^k} \sum_{c \in \{0,1\}^n} \Bbb I \left[c = m M(G), w_H(c)=v \right] \\ \nonumber
&=& \frac {1}{2} \sum_{u=0}^k \sum_{m \in Z^{(k,u)}}  \sum_{c \in Z^{(n,v)} } \Bbb I \left[c = m M(G)\right] \\
&=& \frac 1 2 \sum_{u=0}^k  A_{u,v}(G).
\end{eqnarray}
\end{proof}

Note that $A_{u,v}(G)$ can be seen as an input-output weight distribution of 
an LDGM code \cite{Hu} if we regard $M(G)$ as a generator matrix with column weight 2.
The variables $u$ and $v$ correspond to the input and output weights, respectively.
Note that $A_{u,v}(G)$ also has a close relationship to the coset weight distribution of 
the low-density parity check (LDPC) code corresponding to a sparse parity check matrix with column weight 2 \cite{wadayama2}.

\section{Ensemble average of cut-set weight distribution}
\label{avecut}
In this section, we discuss the  average (i.e. expectation) of $A_{u,v}(G)$ over the ensemble ${\cal G}_{k,n}$.
In the following, the expectation operator ${\sf E}$ is defined as
\begin{equation}
{\sf E}[f(G)] = \sum_{G \in G _{k,n} } P(G) f(G),
\end{equation}
where $f$ is any real-valued function defined on $G_{k,n}$.

As a preparation for deriving ${\sf E}[A_{u,v}(G)]$, we introduce the following lemma:
\begin{lemma}\label{closedlemma}
Assume that  $m^* \in Z^{(k,u)} $ and 
$c^* \in Z^{(n,v)} $, where $u \in [0,k]$ and $v \in [0,n]$. The following equality holds: 
\begin{eqnarray} \label{closed}
{\sf E} \left[ \Bbb I \left[m^* M(G) = c^* \right]  \right]
&=& \frac{1}{{n \choose v} {{k \choose 2}\choose  n}}{u(k-u) \choose v}  
{{k \choose 2} - u(k-u) \choose n - v}.
\end{eqnarray}
\end{lemma}
\begin{proof}
Due to the symmetry of the ensemble,  we can assume without loss of generality that
the first $u$ elements of $m^*$ are one and the remaining elements are zero.
In a similar manner, $c^*$ is assumed to be a binary vector such that
the first $v$ elements are one and the remaining elements are zero.

In the following, we count the number of labeled graphs that satisfy $m^* M(G) = c^*$ by
counting the number of binary incidence matrices satisfying the above condition.
Let 
$
M(G)=  \left( f_1 \  f_2 \ \cdots f_n \right)
$
where $f_i$ is the $i$th column vector of $M(G)$.
Since $M(G)$ is an incidence matrix, the column weight of $f_i$ is $w_H(f_i) = 2$ for $i \in [1,n]$.
From the above assumptions, we obtain
\begin{equation} \label{mfi}
m^* f_i = 
\left\{
\begin{array}{ll}
1, & i \in [1,v] \\
0, & i \in [v+1,n]. \\
\end{array}
\right.
\end{equation}
We then count the number of allowable combinations of $(f_1,f_2,\ldots, f_n)$ that satisfy (\ref{mfi}).
Let 
\begin{equation}
A = \{f \in \{0,1\}^k \mid  m^* f = 1, w_H(f)=2\}.
\end{equation}
The cardinality of $A$ is given by $|A |= u(k-u)$ because 
a non-zero component of $f$ must have an index within $[1,u]$, and
another non-zero component has an index in the range $[u+1, k]$.
This observation leads to the number of possibilities for $(f_1,f_2,\ldots, f_v)$,
which is given by
$
v! {u(k-u) \choose v}.
$
The remaining $n-v$ columns, $(f_{v+1}, \ldots, f_n)$, should be taken from 
the set $\{f \in \{0,1\}^k \mid  w_H(f)=2\} \backslash A$.
Thus, the number of possibilities for such choices is 
$
(n-v)! {{k \choose 2} - u(k-u) \choose n - v}.
$
In summary, the number of allowable combinations of $(f_1,f_2,\ldots, f_n)$ is given by
\begin{equation}\label{cardx}
\sum_{G \in G_{k,n}}   \Bbb I[m^* M(G) = c^*]= v! (n-v)! {u(k-u) \choose v}{{k \choose 2} - u(k-u) \choose n - v}.
\end{equation}
Thus, the left-hand side of (\ref{closed}) can be rewritten as follows:
\begin{eqnarray} \nonumber
{\sf E}[ \Bbb I[m^* M(G) = c^*]  
&=& \sum_{G \in G_{k,n}}  P(G) \Bbb I[m^* M(G) = c^*] \\ \nonumber
&=&\frac{1}{n! {{k \choose 2}\choose  n}} \sum_{G \in G_{k,n}} \Bbb  I[m^* M(G) = c^*] \\ 
&=& \frac{v! (n-v)! }{n! {{k \choose 2}\choose  n}}{u(k-u) \choose v}  {{k \choose 2} - u(k-u) \choose n - v}.
\end{eqnarray}
The final equality is due to  (\ref{cardx}). 
\end{proof}

The primary result in this section is given as follows:
\begin{lemma} \label{auvw}
The expectation of $A_{u,v}(G)$ is given by
\begin{eqnarray} 
{\sf E}[A_{u,v}(G)]
&=& \frac{1}{{{k \choose 2}\choose  n}}
{k \choose u} {u(k-u) \choose v} {{k \choose 2} - u(k-u) \choose n - v}, 
\end{eqnarray}
where $u \in [0,k], v \in [0,n]$.
\end{lemma}
\begin{proof}
The expectation of $A_{u,v}(G)$ can be simplified as follows:
\begin{eqnarray} \nonumber
{\sf E}[A_{u,v}(G)] &=&
{\sf E}  \left[ \sum_{m \in Z^{(k,u)} } \sum_{c \in Z^{(n,v)} }\Bbb I[m M(G) = c]   \right] \\  \nonumber
&=&  \sum_{m \in Z^{(k,u)} } \sum_{c \in Z^{(n,v)} } {\sf E}  \left[\Bbb I[m M(G) = c] \right]\\ \label{aaa}
&=&
{k \choose u} {n \choose v} {\sf E}  \left[ \Bbb I[m^* M(G) = c^*] \right],
\end{eqnarray}
where the final equality is due to the symmetry of the ensemble. The binary vectors $m^* \in Z^{(k,u)} $
$c^* \in \{0,1\}^n $ can be chosen arbitrarily.
Substituting (\ref{closed}) in the previous Lemma into (\ref{aaa}), we obtain the claim of this lemma.
\end{proof}

\section{Probability of unconnected graphs}
\label{probunconnect}
In this section, we discuss the probability such that a randomly chosen graph in $G_{k,n}$ is unconnected.
Such a probability has been investigated in detail \cite{RANDOMGRAPH} \cite{Bollobas}.  
These bounds are required for deriving the primary results of the present paper, which are described later.

\subsection{Upper bound on probability for unconnected graphs}

The upper bound on the unconnected probability 
presented in this section is derived based on the average cut-set weight distribution. 
The unconnected probability $P_U(k,n)$ is defined as 
$
P_U(k,n) = {\sf E} \left[ \Bbb I [G \mbox{: \unc}] \right].
$
An upper bound on the unconnected probability is given by the following lemma:
\begin{lemma} \label{unconnected}
The probability for selecting an unconnected graph from the ensemble ${\cal G}_{k,n}$ is upper 
bounded by 
\begin{eqnarray} 
P_U(k,n) &\le&  
\frac{1}{2{{k \choose 2}\choose  n}} \sum_{u=0}^{k} {k \choose u}  {{k \choose 2}-u(k-u) \choose n  }  - 1.
\end{eqnarray}

\end{lemma}
\begin{proof}
Let $p_i$ be the probability such that the incidence matrix of a randomly chosen graph 
has rank $k-i (i \in [1,k-1])$, i.e., 
\begin{equation}
p_i = \sum_{G \in G_{k,n}} P(G) \Bbb I [{\sf rank}(M(G)) = k-i] .
\end{equation}
Using $p_i$, the unconnected probability $P_U(k,n)$ can be rewritten as follows: 
\begin{eqnarray} \nonumber
P_U(k,n) &=& {\sf E} \left[ \Bbb I [G \mbox{: \unc}] \right] \\ 
&=& \sum_{G \in G_{k,n}} P(G) \Bbb I [{\sf rank}(M(G)) < k-1]  
= \sum_{i=2}^{k-1} p_i.
\end{eqnarray}
The cardinality of the set $\{m \in \{0,1\}^k \mid  m M(G) = 0 \}$ is denoted by $T(G)$, which 
can be transformed into 
\begin{eqnarray} \nonumber
T(G) 
&=& |\{m \in \{0,1\}^k \mid  m M(G)=0  \}|  \\
&=&   \sum_{u=0}^k \sum_{m \in Z^{(k,u)}} \Bbb I [m M(G)=0].
\end{eqnarray}
The expectation of $T(G)$ can be expressed simply as follows:
\begin{eqnarray} \nonumber
{\sf E}[T(G)] 
&=& {\sf E} \left[\sum_{u=0}^k \sum_{m \in Z^{(k,u)}} \Bbb I [m M(G)=0]  \right] \\ \nonumber
&=& \sum_{u=0}^k \sum_{m \in Z^{(k,u)}}{\sf E} \left[ \Bbb I [m M(G)=0]  \right] \\ \nonumber
&=& \sum_{u=0}^k  {k \choose u} {\sf E} \left[ \Bbb I [m M(G)=0]  \right] \\
&=& \sum_{u=0}^k \frac{{k \choose u}  {{k \choose 2}-u(k-u) \choose n  } }{ {{k \choose 2} \choose n} }.
\end{eqnarray}
The final equality is due to a special case of Lemma \ref{closedlemma}.
The expectation $E[T(G)]$ can be lower bounded by $2+2 P_U(k,n)$ as follows: 
\begin{eqnarray} \nonumber
E[T(G)] 
&=& 2 p_1 + 4 p_2 + \cdots 2^{k-1} p_{k-1} \\ \nonumber
&\ge& 2 p_1 + 4(p_2 + \cdots  p_{k-1}) \\ \nonumber
&=& 2 (1-P_U(k,n)) + 4P_U(k,n) \\
&=& 2  + 2P_U(k,n). 
\end{eqnarray}
This lower bound on $E[T(G)]$ leads to an upper bound on $P_U(k,n)$, as shown below:
\begin{eqnarray} \label{PUupper}
P_U(k,n) &\le&  \frac{1}{2{{k \choose 2}\choose  n}}  \sum_{u=0}^{k} {k \choose u}  {{k \choose 2}-u(k-u) \choose n  }  - 1.
\end{eqnarray}
\end{proof}
Note that Erd\H{o}s and R\'{e}nyi \cite{RANDOMGRAPH} derived
a similar bound based on another combinatorial argument.

\subsection{Lower bound on probability for unconnected graphs}
The lower bound on the probability for an unconnected graph presented below can be derived based on 
a simple combinatorial argument.
\begin{lemma} \label{unconnectedlower}
The probability of selecting an unconnected graph from the ensemble ${\cal G}_{k,n}$ is lower 
bounded by 
\begin{equation}
P_U(k,n) \ge
\frac{k}{{{k \choose 2} \choose n}} \left( {{k-1 \choose 2} \choose n} - (k-1){{k-2 \choose 2} \choose n} \right).
\end{equation}
\end{lemma}
\begin{proof}
The main concept of the lower bound is to bound $P_U(k,n)$ from below by the probability 
such that a randomly chosen graph $G$ has  a zero row vector in its incidence matrix $M(G)$, i.e.,
\begin{eqnarray} \nonumber
P_U(k,n) &=& {\sf E} \left[ \Bbb I [G \mbox{: \unc}] \right] \\ \label{aaa}
&\ge& {\sf E} \left[ \Bbb I [M(G) \mbox{ has a zero row vector}] \right].
\end{eqnarray}
This inequality is obtained for the following reason.
If $M(G)$ contains a zero row vector, then $G$ is unconnected because $G$ contains 
an isolated node corresponding to zero vectors of $M(G)$. This implies that
\begin{equation}
\Bbb I [G \mbox{: \unc}]  \ge  \Bbb I [M(G) \mbox{ has a zero row vector}] 
\end{equation}
holds for any $G \in G_{k,n}$. The right-hand side of (\ref{aaa}) can be further simplified as follows:
\begin{eqnarray} \nonumber
{\sf E} \left[ \Bbb I [M(G) \mbox{ has a zero row vector}] \right]   
&=& \sum_{G \in G_{k,n}} P(G) \Bbb I [M(G) \mbox{ has a zero row vector}] \\ \nonumber
&=&\frac{1}{n! {{k \choose 2} \choose n}}  \sum_{G \in G_{k,n}} \Bbb I [M(G) \mbox{ has a zero row vector}] \\ 
&\ge&\frac{n! k}{n! {{k \choose 2} \choose n}}  \left( {{k-1 \choose 2} \choose n} - (k-1){{k-2 \choose 2} \choose n} \right).
\end{eqnarray}
The final equality is due to the following combinatorial argument.  
Here, we count the number of possible binary matrices having a zero vector.
Every column of $M(G)$ must contain two-ones, and there are $k-1$-possible positions for such 
two-ones because we assumed that a row are constrained to be zero. This implies that the size of the set of possible column vectors becomes ${k-i \choose 2}$.  
Therefore, the number of binary matrices having a zero row vector can be lower bounded by
\begin{equation} \label{overcount}
 \sum_{G \in G_{k,n}} \Bbb I [M(G) \mbox{ has a zero row vector}]  \ge
   n! k \left( {{k-1 \choose 2} \choose n} - (k-1){{k-2 \choose 2} \choose n} \right).
\end{equation}
The factors $k$ and $n!$ in the above expression are the number of possible positions of the zero vector and
the number of possible way to sort $n$-column vectors, respectively.
The negative term $- (k-1){{k-2 \choose 2} \choose n}$ compensates the overcounts for multiple zero 
vectors. Again, this negative compensation term  overcounts  multiple zero vectors. 
Therefore, RHS of (\ref{overcount}) is smaller than LHS of (\ref{overcount}).
The claim of the lemma follows directly from the above discussion.
\end{proof}

\section{Bounds on expected network failure probability}

We assume that the edge failures occur independently with probability $\epsilon (0 < \epsilon < 1)$. For a given graph $G \in G_{k,n}$, the edge failures transform $G$ into $G'$, where $G'$ is referred to as a {\em survivor subgraph} of $G$. The network failure probability $P_f(G,\epsilon)$ is the probability such that the survivor subgraph $G'$ becomes unconnected.
The precise definition of the network failure probability is given as follows:
\begin{definition}
For $G \in G_{k,n}$, the network failure probability $P_f(G,\epsilon)$ is defined by
\begin{equation}
P_f(G,\epsilon) 
= \Bbb I[G \mbox{: \con}]\left(\sum_{E' \subset E(G)}  \Bbb I[(V(G),E \backslash E') \mbox{: \unc}] 
\epsilon^{|E'|} (1-\epsilon)^{n - |E'|} \right)+ \Bbb I [G \mbox{: \unc}].
\end{equation}
\end{definition}

Based on this definition, it is evident that $P_f(G,\epsilon) = 1$ holds if $G$ is unconnected.
If $G$ is a connected graph, then $P_f(G,\epsilon)$ is equal to the probability corresponding to unconnected survivor subgraphs of $G$.
In this section, the upper and lower bounds on the expected network failure probability are presented.

\subsection{Upper bound on expected network failure probability}

The following theorem gives an upper bound on the expected network failure probability.
\begin{theorem}[Upper bound]
\label{upperbound}
The expectation of the network failure probability $P_f(G,\epsilon)$ over ${\cal G}_{k,n}$ can be 
upper bounded by 
\begin{equation}
{\sf E}[P_f(G,\epsilon)] \le
\frac{1}{2{{k \choose 2}\choose  n}} \sum_{v=1}^n  \sum_{u=0}^k 
\gam  \epsilon^{v}   
+  \frac{1}{2{{k \choose 2}\choose  n}}  \sum_{u=0}^{k} {k \choose u}  {{k \choose 2}-u(k-u) \choose n  }  - 1.
\end{equation}
\end{theorem}

\begin{proof}
We first derive an upper bound on $P_f(G,\epsilon)$ for any $G \in G_{k,n}$.
\begin{eqnarray} \nonumber
P_f(G,\epsilon) &=&  \Bbb I[G \mbox{: \con}] \left(\sum_{E' \subset E(G)}  \Bbb I[(V(G),E \backslash E') \mbox{: \unc}] 
\epsilon^{|E'|} (1-\epsilon)^{n - |E'|} \right)
+ \Bbb I [G \mbox{: \unc}] \\ \nonumber
&\le& \Bbb I[G \mbox{: connected}] \left(\sum_{v=1}^n 
\sum_{E' \subset E(G)} \Bbb I [E' \mbox{ is a cut-set of }G, \     \omega(E')= v] 
\epsilon^{v} (\epsilon + (1-\epsilon))^{n-v}   \right)+ \Bbb I [G \mbox{: unconnected}] \\ \nonumber
&=& \Bbb I[G \mbox{: \con}] \left(\sum_{v=1}^n  B_v(G)
\epsilon^{v}   \right) + \Bbb I [G \mbox{: \unc}] \\ \nonumber
&=& \Bbb I[G \mbox{: \con}] \left(\frac 1 2 \sum_{v=1}^n  \sum_{u=0}^k A_{u,v}(G)
\epsilon^{v}   \right) + \Bbb I [G \mbox{: \unc}] \\ \label{unc}
&\le& \frac 1 2 \sum_{v=1}^n  \sum_{u=0}^k A_{u,v}(G) \epsilon^{v}  + \Bbb I [G \mbox{: \unc}].
\end{eqnarray}
Lemma \ref{bwg} is used in this derivation.
The final inequality is due to the inequality $\Bbb I[G \mbox{: \con}] \le 1$ for any $G \in G_{k,n}$. 
Using the upper bound shown above, the expectation of $P_f(G,\epsilon)$ can be upper bounded as follows:
\begin{eqnarray} \nonumber
{\sf E}[P_f(G,\epsilon)] &\le& 
{\sf E}\left[ \frac 1 2 \sum_{v=1}^n  \sum_{u=0}^k A_{u,v}(G) \epsilon^{v} \right]  + {\sf E}[\Bbb I [G \mbox{: \unc}]] \\ \nonumber
&=& 
 \frac 1 2 \sum_{v=1}^n  \sum_{u=0}^k {\sf E}\left[ A_{u,v}(G) \right] \epsilon^{v}   + P_U(k,n) \\ \nonumber
&\le& 
\frac{1}{2{{k \choose 2}\choose  n}} \sum_{v=1}^n  \sum_{u=0}^k 
\gam  \epsilon^{v}   
+  \frac{1}{2{{k \choose 2}\choose  n}}  \sum_{u=0}^{k} {k \choose u}  {{k \choose 2}-u(k-u) \choose n  }  - 1. \\
\end{eqnarray}
The first inequality is due to (\ref{unc}). Lemmas \ref{auvw} and \ref{unconnected} are exploited in the derivation.
\end{proof}

\subsection{Lower bound on an expected network failure probability}
The following theorem includes a lower bound on the expected network failure probability.
\begin{theorem}[Lower bound]
The expectation of the network failure probability $P_f(G,\epsilon)$ over ${\cal G}_{k,n}$ can be 
lower bounded by 
\begin{equation}
{\sf E}[P_f(G,\epsilon)] \ge
\frac{1}{2{{k \choose 2}\choose  n}} \sum_{v=1}^n  \sum_{u=0}^k 
{k \choose u} {u(k-u) \choose v} {{k \choose 2} - u(k-u) \choose n - v} \epsilon^{v}   (1-\epsilon)^{n-v}
+ \frac{k}{2{{k \choose 2} \choose n}} \left( {{k-1 \choose 2} \choose n} - (k-1){{k-2 \choose 2} \choose n} \right).
\end{equation}
\end{theorem}
\begin{proof}
As in the case of the upper bound, we start from the definition of $P_f(G,\epsilon)$ and then bound $P_f(G,\epsilon)$ from below:
\begin{eqnarray} \nonumber
P_f(G,\epsilon) &=&\Bbb I[G \mbox{: \con}] \left(\sum_{E' \subset E(G)}  \Bbb I[(V(G),E \backslash E') \mbox{: \unc}] \epsilon^{|E'|} (1-\epsilon)^{n - |E'|}\right)+ \Bbb I [G \mbox{: \unc}] \\ \nonumber
&\ge& \Bbb I[G \mbox{: connected}] \left(\sum_{v=1}^n 
\sum_{E' \subset E(G)} \Bbb I [E' \mbox{ is a cut-set of }G, \     \omega(E')= v]
\epsilon^{v} (1-\epsilon)^{n-v}   \right) + \Bbb I [G \mbox{: unconnected}] \\ \nonumber
&=& \Bbb I[G \mbox{: \con}] \left(\sum_{v=1}^n  B_v(G)
\epsilon^{v} (1-\epsilon)^{n-v}   \right) 
+ \Bbb I [G \mbox{: \unc}] \\ \nonumber
&=& \Bbb I[G \mbox{: \con}]\left(\frac 1 2 \sum_{v=1}^n  \sum_{u=0}^k A_{u,v}(G)\epsilon^{v} (1-\epsilon)^{n-v} \right)
+ \Bbb I [G \mbox{: \unc}] \\ \nonumber
&=& (1-\Bbb I[G \mbox{: \unc}]) \left(\frac 1 2 \sum_{v=1}^n  \sum_{u=0}^k A_{u,v}(G)\epsilon^{v} (1-\epsilon)^{n-v} \right)
+ \Bbb I [G \mbox{: \unc}] \\ \nonumber
&\ge&\left(\frac 1 2 \sum_{v=1}^n  \sum_{u=0}^k A_{u,v}(G)\epsilon^{v} (1-\epsilon)^{n-v} \right)
 - \frac 1 2 \Bbb I [G \mbox{: \unc}] +  \Bbb I [G \mbox{: \unc}] \\
&=&\left(\frac 1 2 \sum_{v=1}^n  \sum_{u=0}^k A_{u,v}(G)\epsilon^{v} (1-\epsilon)^{n-v} \right)
 + \frac 1 2 \Bbb I [G \mbox{: \unc}].
\end{eqnarray}
Lemma \ref{bwg} is used in the above derivation. The inequality based on the binomial theorem
\begin{eqnarray} \nonumber
\frac 1 2 \sum_{v=1}^n  \sum_{u=0}^k A_{u,v}(G) \epsilon^{v} (1-\epsilon)^{n-v} 
&\le& \frac 1 2\sum_{v=0}^n   {n \choose v}\epsilon^{v} (1-\epsilon)^{n-v} 
= \frac 1 2
\end{eqnarray}
is also exploited in the above transformation. By taking the expectation, 
we immediately obtain the lower bound as follows: 
\begin{eqnarray} \nonumber
{\sf E}[P_f(G, \epsilon)] &\ge& 
{\sf E}\left[ \frac 1 2 \sum_{v=1}^n  \sum_{u=0}^k A_{u,v}(G) \epsilon^{v} (1-\epsilon)^{n-v}  + \frac 1 2 \Bbb I [G \mbox{: \unc}] \right]   
\\ \nonumber
&=& 
 \frac 1 2 \sum_{v=1}^n  \sum_{u=0}^k {\sf E}\left[ A_{u,v}(G) \right] \epsilon^{v} (1-\epsilon)^{n-v} + \frac 1 2 P_U(k,n) \\ \nonumber
&\ge& 
\frac{1}{2{{k \choose 2}\choose  n}} \sum_{v=1}^n  \sum_{u=0}^k 
\gam  \epsilon^{v}   (1-\epsilon)^{n-v}  \\
&+& \frac{k}{2{{k \choose 2} \choose n}} \left( {{k-1 \choose 2} \choose n} - (k-1){{k-2 \choose 2} \choose n} \right). 
\end{eqnarray}
For bounding  $P_U(k,n)$ from below,  Lemma \ref{unconnectedlower} has been used.
\end{proof}

\section{Numerical evaluation}
In order to verify the tightness of the bounds on the expected network failure probability proved in the previous section, 
we have performed numerical evaluations for the expected network failure probability.
Figure \ref{fig1st} presents the upper and lower bounds on the expected network failure probability for the case in which 
$k=6$ and $n=12$. 
For the purpose of comparison, the exact values are also plotted in Fig. \ref{fig1st}. These exact values have been obtained by
generating all possible graphs and by evaluating the network failure probability using 
a recursive computation referred to as {\em pivotal decomposition} \cite{IRE}.
The horizontal axis represents the edge failure probability $\epsilon$, and the vertical axis denotes the value of the expected network failure probability. 
Note that the upper and lower bounds are reasonably tight in this case. Both bounds are very close to the exact value when $\epsilon$ is smaller than $10^{-2}$.

Figure \ref{fig2nd} shows the values of upper and lower bounds on the expected network failure probability for the case in which $k=7$ and $n=12$. The curve for the exact probability (solid line) approaches a constant value around 
$10^{-2}$ as $\epsilon$ becomes small. The reason for this floor phenomenon is as follows. The unconnected probability is
$P_U(7,12) = 0.0108$ in this case. If $\epsilon$ is small, e.g., $10^{-3}$, the dominant failure event is to select an unconnected graph from the ensemble. Therefore, the expected network failure probability approaches $P_U(7,12)$.
Note that, in the case of Fig. \ref{fig1st}, we can prove $P_U(6,12)=0$ by showing that the upper bound (\ref{PUupper}) takes a value of zero. Thus, we can observe that the curves in Fig. \ref{fig1st} decrease monotonically as $\epsilon$ decreases.

\begin{figure}[htbp]
\begin{center}
\includegraphics[width=0.5 \linewidth]{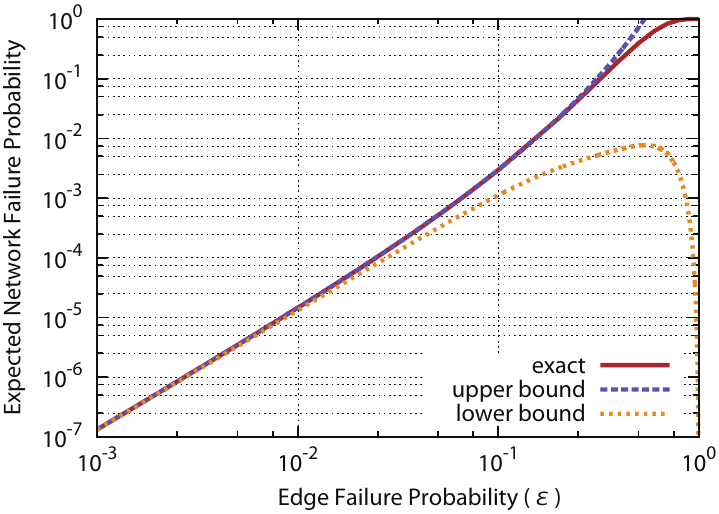} 
\end{center}
	\caption{Expected network failure probability $(k=6, n=12)$: exact, upper and lower bounds}
	\label{fig1st}
\end{figure}
\begin{figure}[htbp]
\begin{center}
\includegraphics[width=0.5 \linewidth]{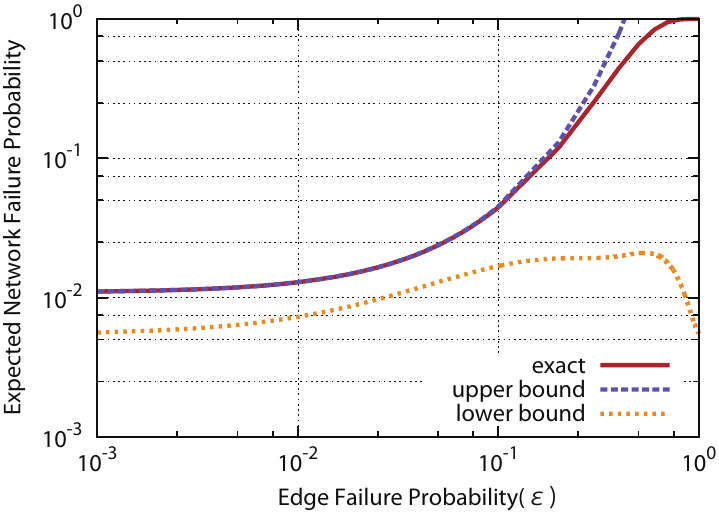} 
\end{center}
	\caption{Expected network failure probability $(k=7, n=12)$: exact, upper and lower bounds}
	\label{fig2nd}
\end{figure}

\section{Conclusion}

In the present paper, upper and lower bounds on the expected network failure probability are derived. The average cut-set weight distribution is key in deriving these bounds. 
The ensemble analysis used here is similar to the analysis of the input-output weight distribution of an LDGM code with column weight 2.
The present study reveals a close relationship between the ensemble analysis of the network failure probability and the ensemble analysis of the error detection probability \cite{wadayama} of LDGM codes with column weight 2. This link between the network reliability problem 
and coding theory may provide a new perspective on the network reliability problem.

In the present study, we focused solely on the upper and lower bounds for fixed $n$ and $k$.  
It would be interesting to investigate the asymptotic behavior of the network failure probability when $n$ and $k$ approach infinity while maintaining the relationship $n = f(k)$ ($f$ is a real-valued function, e.g., $n = \beta k^2$). The error exponent analysis for the undetected error probability of LDPC  codes  shown in \cite{wadayama} is a possible choice for such an asymptotic analysis.

In the present paper, we discussed a simple graph ensemble, which is closely related to 
the Erd\H{o}s-R\'{e}nyi random graph ensemble \cite{RANDOMGRAPH}. In such a graph ensemble, the degree of a node is strongly concentrated around its expectation when $n$ and $k$ become large. It seems natural to consider the behavior of the network failure probability for a graph ensemble with a non-uniform degree distribution. The analysis for deriving the average cut-set weight distribution is approximately equivalent to the analysis of the average weight distribution of an LDGM code ensemble \cite{Hu} or of the average coset weight distribution of an LDPC code ensemble. Some known results, such as those reported in \cite{wadayama2}, may be exploited for the analysis of an ensemble with non-uniform degree distribution. The average cut-set weight distribution derived in the present paper may also be useful for several applications other than the evaluation of the expected network failure probability. For example, statistical information on the max-flow of a randomly chosen graph could be obtained from the average cut-set weight distribution in combination with the min-cut max-flow theorem.

\end{document}